\PassOptionsToPackage{colorlinks=false,urlcolor=black,linkcolor=black,citecolor=black}{hyperref}

\documentclass[11pt]{eptcs}  %

\usepackage{breakurl}

\providecommand{\event}{MCU 2013} %
\DeclareMathAlphabet{\mathcal}{OMS}{cmsy}{m}{n}
  \newenvironment{IEEEkeywords}
    {\textbf{Keywords:}}
    {}
  \newcommand{\appropriatebib}{\bibliographystyle{eptcs}}

\usepackage{amsthm}
\usepackage{amssymb}
\usepackage[]{amsmath}  %

\usepackage{paralist}  %

\ifx\nialldrafting\undefined
  \usepackage{tikz}
  \usepackage{verbatim}
   \usepackage[colorinlistoftodos]{todonotes}%
   \usepackage{xcolor}
   \newcommand{\nmnoteInline}[1]{} %
   \newcommand{\dw}[1]{}  %
\else
    \newcommand{\dw}[1]{\todo[inline, color=blue!30]{#1}}

    \newcommand{\nmnoteInline}[1]{\todo[inline, color=red!30]{#1}}

\fi
\usetikzlibrary{positioning,matrix}
\usetikzlibrary{automata}
\usetikzlibrary{decorations.markings}
\usetikzlibrary{arrows}
\usetikzlibrary{calc}
\usetikzlibrary{shapes}

\usepackage[utf8]{inputenc}
\usepackage[english]{babel}

\usepackage[]{hyperref} %

\usepackage{microtype}   %

\usepackage[classfont=bold,funcfont=roman]{complexity}

\newtheorem{theorem}{Theorem}
\newtheorem{definition}[theorem]{Definition}  %

\newtheorem{lemma}[theorem]{Lemma}
\newtheorem{corollary}[theorem]{Corollary}

\newcommand{\ORgate}{\textsc{Or}}
\newcommand{\ANDgate}{\textsc{And}}
\newcommand{\NOTgate}{\textsc{Not}}
\newcommand{\XORgate}{\textsc{Xor}}
\newcommand{\ORACLEgate}{\textsc{Oracle}}

\newclass{\DLOGTIME}{DLOGTIME}
\newclass{\FL}{FL}
\newclass{\FAC}{FAC}
\newclass{\FNC}{FNC}
\newclass{\tally}{tally}
\newclass{\length}{length}
\newclass{\COR}{OR}
\newclass{\CAND}{AND}
\newclass{\CANDCOR}{\CAND\textrm{-}\COR}

\newlang{\Parity}{Parity}

\newcommand{\Reducible}[3]{\ensuremath{#2_{\mathrm{#1}}(#3)}}

\newcommand{\bigO}[1] {\ensuremath{\mathcal{O}(#1)}}

\newcommand{\tallyMachine}{\ensuremath{\mathcal{M}}}

\newcommand{\semiencoder}{\ensuremath{h}}
\newcommand{\former}{\ensuremath{f}}

\newcommand{\uniform}      [2]{\ensuremath{#1\textrm{-uniform-}#2}}
\newcommand{\semiunif}      [2]{\ensuremath{#1\textrm{-semi-uniform-}#2}}
\newcommand{\uniformAND}    [1]{\ensuremath{#1\textrm{-uniform-}\CAND}}
\newcommand{\uniformOR}    [1]{\ensuremath{#1\textrm{-uniform-}\COR}}
\newcommand{\semiuniformOR}[1]{\ensuremath{#1\textrm{-semi-uniform-}\COR}}
\newcommand{\semiuniformAND}[1]{\ensuremath{#1\textrm{-semi-uniform-}\CAND}}
\newcommand{\FOSOR}{\ensuremath{\semiuniformOR{\FAC^0}}}  
\newcommand{\FOSAND}{\ensuremath{\semiuniformAND{\FAC^0}}}  
\newcommand{\FOOR}{\ensuremath{\uniformOR{\FAC^0}}} 
\newcommand{\FOAND}{\ensuremath{\uniformAND{\FAC^0}}}

\title{AND and/or OR: \\ Uniform Polynomial-Size Circuits}
\def\titlerunning{AND and/or OR: Uniform Polynomial-Size Circuits}
\def\authorrunning{N. Murphy and D. Woods}
  \author{%
    Niall Murphy%
      \thanks{N.\ Murphy was supported by
      the Campus de Excelencia Internacional Moncloa UCM-UPM PICATA Program
      and the Irish Research Council for Science, Engineering and Technology's EMBARK Initiative.}%
      \institute{Facultad de Inform\'atica\\
       Universidad Polit\'ecnica de Madrid\\ %
       CEI-Moncloa UCM-UPM\\
       Madrid, Spain\\
       \url{niall.murphy@upm.es}}
    \and
    Damien Woods%
      \thanks{D.\ Woods was supported by National Science Foundation (United
        States) grants CCF-1219274, CCF-1162589, and 0832824 (The Molecular
        Programming Project), and Junta de Andalucía (Spain) grant TIC-581.}%
     \institute{Computer Science\\
       Center for Mathematics of Information \\
       California Institute of Technology\\
       Pasadena, CA 91125, USA\\
       \url{woods@caltech.edu}}
    }

\begin{document}
  \maketitle

  \begin{abstract} 
    We investigate the complexity of uniform $\ORgate$ circuits and $\ANDgate$
    circuits of polynomial-size and depth. As their name suggests,  $\ORgate$
    circuits have $\ORgate$ gates as their computation gates, as well as the
    usual input, output and constant $0/1$ gates.  As is the
    norm for Boolean circuits, our circuits have multiple sink gates, which
    implies that an $\ORgate$ circuit computes an $\ORgate$ function on some
    subset of its input variables. Determining that subset amounts to solving a
    number of reachability questions on a polynomial-size directed graph (which
    input gates are connected to the output gate?), taken from a very sparse
    set of graphs. However, it is not obvious whether or not this (restricted) reachability problem can be solved, by say,
    uniform $\AC^0$ circuits (constant depth, polynomial-size, $\ANDgate,
    \ORgate, \NOTgate$ gates). This is one reason why characterizing the power
    of these simple-looking circuits in terms of uniform classes turns out to be
    intriguing. Another is that the model itself seems particularly natural
    and worthy of study.

    Our goal is the systematic characterization of  uniform polynomial-size
    $\ORgate$ circuits, and $\ANDgate$ circuits, in terms of known uniform
    machine-based complexity classes. In particular, we consider the languages
    reducible to such uniform families  of $\ORgate$ circuits, and $\ANDgate$ circuits, 
    under a variety of reduction types.   We give
    upper and lower bounds on the computational power of  these language classes. 
      We find
    that these complexity classes are closely related to $\tally\NL$, the set
    of unary languages within $\NL$, and to sets reducible to $\tally\NL$.
    Specifically, for a variety of types of reductions (many-one, conjunctive
    truth table, disjunctive truth table, truth table, Turing) we give
    characterizations of languages reducible to $\ORgate$ circuit classes in
    terms of languages reducible to $\tally\NL$ classes. Then, some of these
    $\ORgate$ classes are shown to coincide, and some are proven to be
    distinct. We give analogous results for $\ANDgate$ circuits. Finally, for
    many of our $\ORgate$ circuit classes, and analogous $\ANDgate$ circuit
    classes, we prove whether or not the two classes coincide, although we
    leave one such inclusion open.
  \end{abstract}

  \begin{IEEEkeywords}
    Computational complexity; uniform Boolean circuits; AND circuits; OR
    circuits; NL; AC$^0$
  \end{IEEEkeywords}

  \section{Introduction}
  We look at the complexity of simple problems: those defined by uniform
  $\ORgate$ circuits and $\ANDgate$ circuits of polynomial-size and depth.
  As their name suggests, $\ORgate$ circuits have only $\ORgate$ gates as their
  computation gates, as well as the usual input gates, constant ($0/1$) gates,
  and an output gate.  As is the norm for Boolean circuits, our circuits have
  multiple sink gates, which implies that an $\ORgate$ circuit computes an
  $\ORgate$ function on some \emph{subset} of its input variables.
  Determining that subset amounts to solving a number of reachability questions on a
  polynomial-size directed graph (i.e.\ which input gates are connected to the output
  gate?), taken from a very sparse set of graphs. It is not obvious whether or not these reachability questions 
  can be solved, in say, uniform $\AC^0$. Yet these problems are trivially
  in non-uniform-$\AC^0$. This is one reason why characterizing the power of
  these simple-looking circuits in terms of uniform classes turns out to be
  intriguing. Another is that the model itself seems particularly natural and
  worthy of study.

  Our goal is the systematic characterization of polynomial-size uniform
  $\ORgate$ circuits, and $\ANDgate$ circuits, in terms of known uniform
  machine-based complexity classes. In particular, we consider the languages
  reducible to such circuit classes, under a variety of reductions. We give
  upper and lower bounds on the computational power of these classes.
  We find that they are closely related to $\tally\NL$, the set of unary
  languages within $\NL$, and to sets reducible to $\tally\NL$. Specifically,
  for a variety of types of reductions ($\AC^0$ many-one, conjunctive
  truth-table, disjunctive truth-table, truth-table, Turing) we give
  characterizations of languages reducible to $\ORgate$ circuit classes in
  terms of languages reducible to $\tally\NL$ classes. Two of the $\ORgate$
  classes are shown to coincide, and others are proven to be distinct. We give
  analogous results for $\ANDgate$ circuits. Finally, for many of our $\ORgate$
  circuit classes, and analogous $\ANDgate$ circuit classes, we prove whether
  or not the two classes coincide, although we leave one such inclusion open.
  These results are summarized in Figure~\ref{fig:summaryofallresults}.

  We also look at a related notion called semi-uniformity where the uniformity
  function for a circuit family gets access to the input word (and not merely
  its length). For sufficiently weak uniformity functions, this notion is
  analogous to a reduction to a circuit value problem, and there is a very
  simple proof that uniformity is a strictly weaker notion than
  semi-uniformity.  Although not
  covered in this paper, these ideas can be used in an analogous proof that
  semi-uniformity is strictly stronger than uniformity in a model called
  membrane systems~\cite{MW2010c}, answering an open question in that
  field~\cite{Pau2005c}, but which is much simpler to state and prove here in the setting of
  Boolean circuits.

  The paper is structured as follows.
  We begin with basic definitions and results in Sections~\ref{sec:defs} and~\ref{sec:tallynl}. %
  Section~\ref{sec:or} contains our main results
      on characterizing the power of polynomial-size uniform $\ORgate$ circuits.
    We give lower and upper bounds, or characterizations,
      of the complexity classes defined by $\ORgate$ circuits under various kinds of reductions.
    Specifically, we show that polynomial-size uniform $\ORgate$ circuits contain $\tally\NL$
      and are properly contained in $\Reducible{dtt}{\FAC^0}{\tally\NL}$,
      i.e.\ the class of languages $\AC^0$ disjunctive truth-table reducible to $\tally\NL$.
  We go on to show that the following three classes coincide:
    languages many-one reducible,
    and disjunctive truth-table reducible, to uniform $\ORgate$ circuits,
    and the class $\Reducible{dtt}{\FAC^0}{\tally\NL}$.
  These results are shown on the left hand side of Figure~\ref{fig:summaryofallresults}.
  Analogous results for $\ANDgate$ circuits are shown on the right of the same figure and are presented in Section~\ref{sec:and}.
  Results on semi-uniformity are given in Section~\ref{sec:uni-semi}.

  Since we are working with extremely weak classes it is important to use
  appropriate reductions between problems and appropriate uniformity
  requirements on circuits.
  We use $\DLOGTIME$-uniform $\FAC^0$~\cite{BIS1990p}
    for reductions~\cite{AMBCDR2009,AK2010,All2012} and
    circuit uniformity~\cite{Agr2001c,Agrawal2011}, 
    which is powerful enough to implement a variety
    of encoding/decoding functions, 
    but yet suitable for use with our (weak) classes.

  \newclass{\CC}{CC}
  \newcommand{\MODgate}{\textsc{Mod}}

  One way to think about uniform $\ORgate$
  circuits is that they compute the $\ORgate$ function on a subset of
  $n$ input variables, that subset being defined via a number of directed graph connectivity 
  questions that are implicitly encoded by the uniformity condition. 
  The seemingly simpler $\ORgate$ function on all $n$ variables is trivially in 
  depth 1 uniform $\AC^0$, yet there are unanswered questions there too.
  For example,  
   it is not known if the  $\ORgate$ function on all $n$ variables
    (or indeed the $\ANDgate$ function) is in $\CC^0[q]$, 
    the class of problems accepted by
      constant depth polynomial-size circuits that use $\MODgate_q$ gates~\cite{HK2010}.
  Figure~\ref{fig:summaryofallresults} suggests a number of open questions.
  Are there other classes that can be used to give a tighter characterization of the class of
  problems solved by polynomial-size uniform $\ORgate$ circuits
  ($\uniformOR{\FAC^0}$)? Also,  $\uniformAND{\FAC^0}$? 
  Is there a language in
  $\uniformOR{\FAC^0}$ that is not in %
  $\Reducible{m}{\FAC^0}{\tally\NL}$?
  It would be interesting to look at the
  power of uniform polynomial-size circuits consisting of other, apparently
  weak, gates, such as $\XORgate$.
  Ko~\cite{Ko1989} shows that
  the classes of languages polynomial time disjunctive and conjunctive
  reducible to $\tally$ are distinct. 
  If it is possible to apply Ko's technique, or something like it, to our much more restrictive setting (i.e.\ $\AC^0$ disjunctive/conjunctive reducible to $\tally\NL$), this would show that the four classes $\AC^0$ many-one, disjunctive truth-table, conjunctive truth-table, and truth-table reducible to $\tally\NL$ are in fact distinct, which would in turn clarify the relationship between the $\ORgate$ and   $\ANDgate$ classes that we consider.

    \newcommand{\pgfextractangle}[3]{%
        \pgfmathanglebetweenpoints{\pgfpointanchor{#2}{center}}
                                  {\pgfpointanchor{#3}{center}}
        \global\let#1\pgfmathresult  
    }

  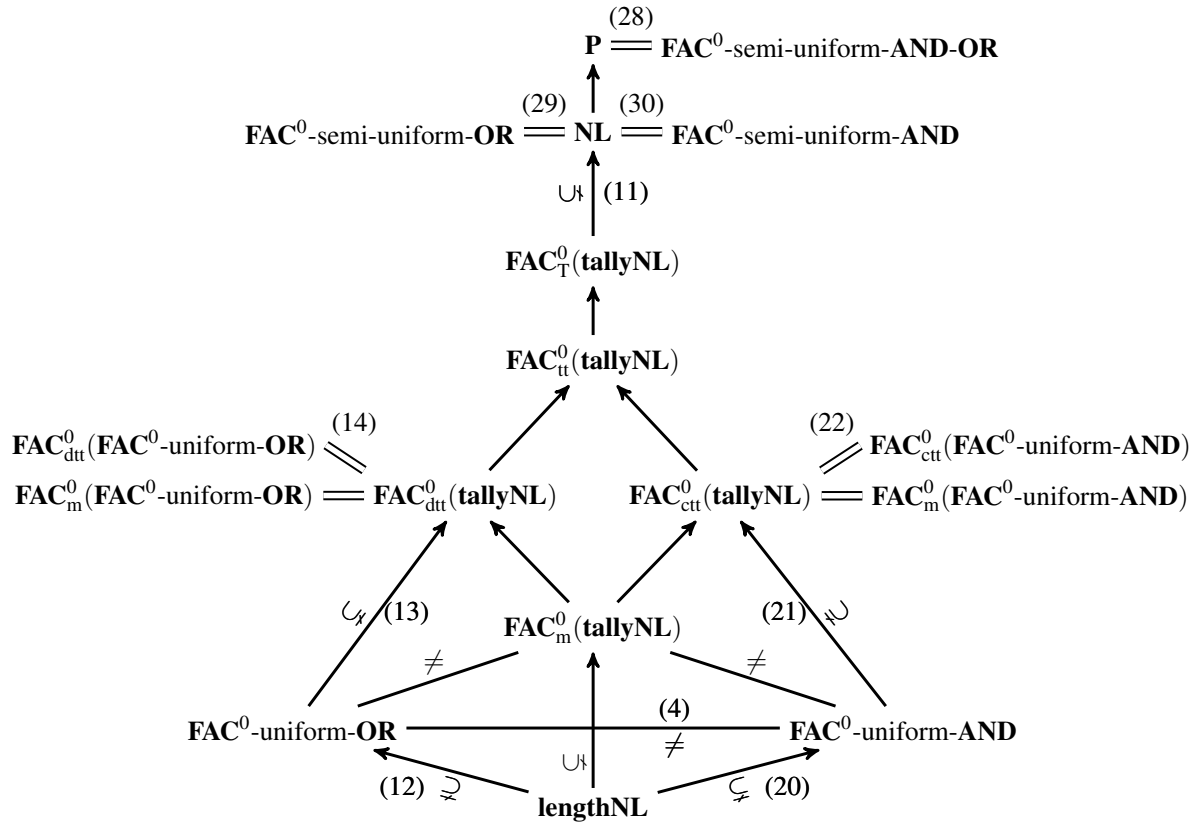
\begin{figure}[t]
    \centering
    \resizebox{\linewidth}{!}{%
    \begin{tikzpicture}[%
      equiv/.style={%
        double distance=0.45ex,double,thick,
        -,
        >=stealth',
      },
      noteq/.style={
        very thick,
        >=stealth',
        -,
        decoration={%
          markings,
          mark=at position 0.5 with {%
            \node [#1] {$\neq$};
          }
        },
        postaction={decorate}
      },
      strictf/.style={
        very thick,
        ->,
        draw,
        >=stealth',
        decoration={%
          markings,
          mark=at position 0.5 with {%
            \node [#1] {$\supsetneq$};
          }
        },
        postaction={decorate}
      },
      strict/.style={
        very thick,
        ->,
        draw,
        >=stealth',
        decoration={%
          markings,
          mark=at position 0.5 with {%
            \node [#1] {$\subsetneq$};
          }
        },
        postaction={decorate}
      },
      refer/.style 2 args={
        decoration={%
          markings,
          mark=at position 0.5 with {%
            \node [inner sep=2pt,#2] {(\ref{#1})};
          }
        },
        postaction={decorate}
      },
      include/.style={->,>=stealth',very thick},
      conjecture/.style={dotted,<->,>=stealth'},
      thm_missing/.style={color=darkgray}]

     \node (tallyNL) {$\length\NL$};
     \node (mTallyNL) [above=5em of tallyNL] {$\Reducible{m}{\FAC^0}{\tally\NL}$};
     \node (ttTallyNL) [above=8em of mTallyNL] {$\Reducible{tt}{\FAC^0}{\tally\NL}$};
     \node (TTallyNL) [above=4ex of ttTallyNL] {$\Reducible{T}{\FAC^0}{\tally\NL}$};
     \node (nl) [above=7ex of TTallyNL] {$\NL$};
     \node (p) [above=4ex of nl] {$\P$};
     \node (semiAND_OR) [right=1.5em of p] {$\semiunif{\FAC^0}{\CANDCOR}$};
     \node (fosor)  [left=1.5em of nl] {$\FOSOR$};
     \node (fosand) [right=1.5em of nl] {$\FOSAND$};

     \node  (uniformOR)  [left=10ex of tallyNL, yshift=3em] {$\uniformOR{\FAC^0}$};
     \node  (uniformAND) [right=10ex of tallyNL,yshift=3em] {$\uniformAND{\FAC^0}$};

     \node (dttTallyNL)    [above=3em of mTallyNL, xshift=-4.7em]  {$\Reducible{dtt}{\FAC^0}{\tally\NL}$};
     \node (cttTallyNL)    [above=3em of mTallyNL, xshift=4.7em] {$\Reducible{ctt}{\FAC^0}{\tally\NL}$};
     \node (mUniformOR)    [left=1.5em of dttTallyNL]  {$\Reducible{m}{\FAC^0}{\uniformOR{\FAC^0}}$};
     \node (mUniformAND)   [right=1.5em of cttTallyNL] {$\Reducible{m}{\FAC^0}{\uniformAND{\FAC^0}}$};
     \node (dttUniformOR)  [above=0em of mUniformOR]  {$\Reducible{dtt}{\FAC^0}{\uniformOR{\FAC^0}}$};
     \node (cttUniformAND) [above=0em of mUniformAND] {$\Reducible{ctt}{\FAC^0}{\uniformAND{\FAC^0}}$};
 
      \draw [equiv={yshift=-1.4ex},refer={lem:semi-mono-or-NL}{yshift=2.2ex}] (fosor) -- (nl); 
      \draw [equiv={yshift=-1.4ex},refer={lem:semi-mono-and-NL}{yshift=2.2ex}] (nl) -- (fosand);
      \pgfextractangle{\angle}{TTallyNL}{nl}
      \draw [strict={xshift=-1.9ex,rotate=\angle},refer={lem:T_reduce_tallyNL_subneq_NL}{xshift=2.7ex}] (TTallyNL) -- (nl); 
      \draw [include] (nl) -- (p); 

      \draw [include] (dttTallyNL) -- (ttTallyNL); %
      \draw [include] (cttTallyNL) -- (ttTallyNL);

      \pgfextractangle{\angle}{tallyNL}{mTallyNL}
      \draw [strict={yshift=-3.5ex,xshift=-1.6ex,rotate=\angle}] (tallyNL) -- (mTallyNL);
      \draw [include] (ttTallyNL) -- (TTallyNL);
      
      \pgfextractangle{\angleb}{tallyNL}{uniformOR}
      \draw [strictf={yshift=-1.4ex,rotate=(\angleb-180)},refer={thm:tallyNL-in-uni-OR}{yshift=-1.4ex,xshift=-1.8em}] (tallyNL) -- (uniformOR);
      
      \pgfextractangle{\anglec}{uniformOR}{dttTallyNL}
      \draw [strict={xshift=-2ex,rotate=\anglec},refer={thm:uniformOR_subneq_dttTallyNL}{xshift=2.4ex}] (uniformOR) -- (dttTallyNL);
      
      \pgfextractangle{\angled}{tallyNL}{uniformAND}
      \draw [strict={yshift=-1.4ex,rotate=\angled},refer={thm:tallyNL-in-uni-AND}{yshift=-1.4ex,xshift=1.8em}] (tallyNL) -- (uniformAND); 
      
      \pgfextractangle{\anglee}{uniformAND}{cttTallyNL}
      \path
      [strictf={xshift=2ex,rotate=(\anglee-180)},refer={lem:uniformAND_subneq_cttTallyNL}{xshift=-2.4ex}] (uniformAND) -- (cttTallyNL);

      \draw [equiv={xshift=0.4ex,yshift=-1.6ex}] (mUniformOR.east) -- (dttTallyNL.west);
      \draw [equiv={xshift=1ex,yshift=1.6ex}] (dttUniformOR.east) -- (dttTallyNL.north west);  
      \node at ($ (dttUniformOR.east) +(0.4,0.3) $) {(\ref{lem:theORcluster_eq})};  %
      \draw [equiv={xshift=-1.5ex,yshift=1ex}] (cttTallyNL.north east) -- (cttUniformAND.west);
      \draw [equiv={xshift=-1ex,yshift=-1ex}] (cttTallyNL.east) -- (mUniformAND.west);
      \node at ($ (cttUniformAND.west) + (-0.4,0.3) $) {(\ref{lem:theANDcluster_eq})};
      \draw [include] (mTallyNL) -- (cttTallyNL);
      \draw [include] (mTallyNL) --  (dttTallyNL);

      \draw [noteq={yshift=-1.4ex,xshift=3em},%
              refer={lem:OR_noteq_AND}{yshift=1.5ex,xshift=3em}]%
              (uniformAND)  -- (uniformOR);

      \draw [noteq={yshift=1.4ex,xshift=-1.6}] (uniformOR)  -- (mTallyNL);
      \draw [noteq={yshift=1.4ex,xshift=1.6}] (uniformAND) -- (mTallyNL);

      \draw [equiv={xshift=-1.7ex},refer={thm:semiuniformCircuitsP}{yshift=2.2ex}] (p) -- (semiAND_OR);

    \end{tikzpicture} 
    } %
    \caption{Summary of results.
      The left side shows relationships between uniform polynomial-size
      $\ORgate$ circuit languages, $\tally\NL$ %
        and sets reducible to these classes. %
      The right side shows analogous relationships for $\ANDgate$ circuit classes. %
      $\Reducible{r}{\FAC^0}{\C}$ denotes $\AC^0$ computable reductions
        of type $\mathrm{r}$ to a class~$\C$.
      Numerical labels refer to theorem statements, and symbols are used to show inclusion type,
        with an unlabelled arrow denoting $\subseteq$.
        To save space, Theorems~\ref{lem:theTuring_cluster_eq} and~\ref{lem:theTruthTable_cluster_eq} are not shown.}
    \label{fig:summaryofallresults}
  \end{figure}

  \section{Definitions}
  \label{sec:defs}
  We now give some basic  definitions based on those in the
  literature~\cite{AK2010, AJ1993, GHR1995x}.  For more details
  on Boolean circuits see~\cite{Vol1999x}.

  For a function $f \colon \{0,1\}^* \to \{0,1\}^*$ and integers $m,n \geq 1$
  let $f_n \colon \{0,1\}^n \to \{ 0,1\}^{m}$ be the restriction of $f$ to
  domain and range consisting  of strings of length $n$ and $m$ respectively
  (we consider only functions~$f$ where for each $n$ there is an $m$ where all
  length-$n$ strings in $f$'s domain are mapped to length-$m$ strings, thus~$f = \bigcup_n f_n$).  

  A \emph{circuit} on $n$ variables $w_0,\ldots,w_{n-1}$ is a directed acyclic
  multi-graph (there may be multiple edges, or wires, between vertices---useful for oracle gates).  The
  vertices of the circuit are generally referred to as gates.  The in-degree
  (out-degree) of a gate is called its fan-in (fan-out).  Each source vertex
  (fan-in 0) is labelled either by one of the input variables
  $w_0,\ldots,w_{n-1}$ or by a constant ``0'' or ``1'' (false or true).  Each
  non-source vertex is labelled by a function name, such as $\ANDgate$,
  $\ORgate$, $\NOTgate$, or $\ORACLEgate$. 
  
  In this paper, we use $\ORACLEgate$ gates. For a given circuit $C$, 
  it will be the case that all $\ORACLEgate$ gates in $C$ compute exactly the same Boolean
  function $g \colon \{0, 1\}^n \to \{0, 1\}$ for $n > 1$, although of course their inputs may be different.  
  We are using the
  following conventions for circuits with tally oracles.  The tally alphabet is
  $\{1\}$. 
  A \emph{tally oracle gate} with $n$ ordered input wires,  takes a string
  of the form $0^{n-i}1^{i}$,  $0 \leq i \leq n$ (encoding the unary word $1^i$) as input,
  and outputs a single bit. 

  Gates with fan-out of 0 (called sinks) may or may not be  
  designated as \emph{output} gates. 

  Given an input $w \in \{0, 1\}^n$, one can inductively assign a Boolean value
  to each vertex of a circuit as follows: each source (\emph{input}) vertex labelled
  by an input variable gets the value of that variable, each  source
  (\emph{constant}) vertex labelled by a constant gets the value of that constant,
  and each other vertex gets the value of the function that labels it applied to
  the values of its children.  Incoming and outgoing edges to a vertex are
  assumed to be ordered (for oracle gates).

  The \emph{depth} of a circuit is the length of the longest path from an input
  vertex to an output vertex. The \emph{size} of a circuit is the number of wires
  it contains~\cite{AK2010}.
  A circuit computes a function on a fixed number of Boolean variables.  We
  consider functions of an arbitrary  number of variables by defining (possibly
  infinite) families of circuits.  We say a family of circuits~$\mathcal{C} =
  \{ C_n \mid n \in \mathbb{N} \}$ computes a function~$f \colon \{0, 1\}^* \to
  \{0, 1\}^*$ if for all $n \in \mathbb{N}$, and for all $w\in\{ 0,1\}^n$
  circuit  $C_n$ outputs the string $f(w)$ (we consider only functions $f$
  where for each $n$ there is an $m$ where all length-$n$ strings in $f$'s
  domain are mapped to length-$m$ strings).  We say a family of circuits
  $\mathcal{C}$ decides a language~$L \subseteq \{ 0,1\}^*$ if for each~$w\in\{0,1\}^n$ circuit~$C_n
  \in \mathcal{C}$ on input $w$ outputs~$1$ if~$w \in L$ and~$0$ if~$w \notin
  L$.
  
  In a \emph{non-uniform} family of circuits there is no required  
   similarity between family members.    
  In order to specify such a requirement   
  we use a \emph{uniformity function} that 
  algorithmically specifies the similarity between members of a circuit family. 
  Roughly speaking, a \emph{uniform circuit family}~$\mathcal{C}$ is an infinite sequence of
  circuits with an associated function~$\former : \{1\}^* \rightarrow
  \mathcal{C}$ that generates members of the family and is
  computable within some resource bound.  
 More precisely:  

  \begin{definition}[$\C$-Uniform circuit family]
    Let $\C$ be a set of functions. A circuit family $\mathcal{C}$ is
    $\C$-uniform, if there is function $f \in \C$, $\former : \{1\}^* \rightarrow
  \mathcal{C}$, where $f(1^n) = C_n$ for all $n \in \mathbb{N}$, 
  and $C_n \in  \mathcal{C}$ is a description of
  a circuit with $n$ input gates (we use $C_n$ to denote
  either a circuit or its encoding as a binary string).
  \end{definition}

  When dealing with uniformity for small complexity classes one of the 
  preferred uniformity conditions is $\DLOGTIME$-uniformity~\cite{BIS1990p}.
  This definition uses an %
  ordering on wires that leave and enter a given gate.

  \begin{definition}[\cite{AK2010}]
    \label{def:dlogtime}
    A circuit family $\mathcal{C}$ is $\DLOGTIME$-uniform if there is
    a procedure that on input $(n, i, r, j, s, t)$, where $n, i, r, j,  s \in
    \mathbb{N}$ are encoded in binary and $t$ is a gate type (e.g., $\ANDgate$,
    $\ORgate$, $\NOTgate$,  input, 0, 1) encoded in binary, runs in time
    \emph{linear} in its input size and accepts if and only if the gate of
    $C_n$ having label $i$ is of type~$t$ and its $r$-th child is the $s$-th
    output  of the gate having label $j$.  In the case where gate $i$ is an input
    gate, the procedure accepts if gate $i$ takes the value of the $s$-th input
    bit.  Furthermore, the procedure accepts inputs of the form $(n, i, j, s,
    output)$ if and only if the $s$-th output wire of gate $i$ is the $j$-th
    output gate %
    of the circuit $C_n$.  We also require that the procedure
    accepts the input $(n, i, d)$ if and only if $d$ is equal to the fan-in of
    the gate of $C_n$ having label $i$. 
  \end{definition}
  
  $\AC^0$ is the set of languages decidable by constant-depth polynomial-size
    (in input length~$n$) \linebreak $\DLOGTIME$-uniform circuits built using unbounded
    fan-in $\ANDgate$ and $\ORgate$ gates, and $\NOTgate$ gates with fan-in 1.
  $\FAC^0$ is the class of functions computable by polynomial-size
    constant-depth $\DLOGTIME$-uniform circuits built using unbounded fan-in
    $\ANDgate$ and $\ORgate$ gates, and $\NOTgate$ gates with fan-in 1.

  An \emph{$\ORgate$ circuit} is a circuit that uses only disjunctive
  logic, that is, a circuit that has only $\ORgate$, constant, and input
  gates. One of the $\ORgate$ gates is denoted as the output gate.
  Similarly an \emph{$\ANDgate$ circuit} is a circuit that uses only conjunctive
  logic, that is, a circuit that has only $\ANDgate$, constant, and input
  gates. One of the $\ANDgate$ gates is denoted as the output gate.
  Note that $\ORgate$  and $\ANDgate$ circuits may have multiple non-output sinks.
    Let non-uniform-$\COR$ (non-uniform-$\CAND$) be the set of decision problems that solved by non-uniform families of
  $\ORgate$ ($\ANDgate$) circuits.
  
  In this paper, we are concerned with $\FOOR$: the class of languages solved
  by uniform polynomial size  $\ORgate$ circuits, formally defined as follows.
  \begin{definition}
    Let $\FOOR$ be the set of decision problems over the 
    alphabet $\{ 0,1\}$ that are solved by $\FAC^0$ uniform families of $\ORgate$ 
    circuits. 
  \end{definition}
  The class $\FOAND$ is defined analogously, but using $\ANDgate$ instead of $\ORgate$ circuits.

  \begin{lemma}
    \label{lem:OR_noteq_AND}
    $\uniformOR{\FAC^0} \neq \uniformAND{\FAC^0}$.
  \end{lemma}
  \begin{proof} 
    An $\ORgate$ circuit computes an $\ORgate$ function on some subset of its
    inputs; in general there is no $\ANDgate$ circuit that computes the same function, and
    vice-versa.  
  \end{proof}

  $\NL$ is the class of languages accepted  
  by non-deterministic logarithmic-space Turing machines. Such machines have a 
  read-only input tape, a write-only output tape and a read-write work tape 
  whose length is a logarithmic function of  input length. 
  The class of functions $f: \{0,1\}^* \to \{0,1\}^*$
  computed by non-deterministic logarithmic-space Turing machines (with an
  additional write-only output tape) is denoted $\FNL$.
  Let $\tally$ be the set of all languages over the one-letter alphabet $\{ 1\}$.
  Let $\length$ be the set of all languages $L \subseteq \{ 0, 1\}^*$ such that
    if $w \in L$ then all words in $\{0,1\}^{|w|}$ are in $L$.

  We define $\tally\NL = \tally \cap \NL$, i.e.\ the class of all tally
  languages and length encoded languages in $\NL$.  
  Let $\tally\co\NL = \tally \cap \co\NL$.
  The following lemma follows from $\NL = \co\NL$,
    (i.e.\ let $L \in \tally\NL \subsetneq \NL = \co\NL$,
    then $L \in \co\NL$ implies $L \in \tally\co\NL$;
    a similar argument holds for the converse):
  \begin{lemma}
    \label{lem:tallynl_eq_tallyconl}
    $\tally\NL = \tally\co\NL$
  \end{lemma}
    Let $\length\NL =  \length \cap \NL $ and $\length\co\NL =  \length \cap \co\NL $.  Also $\length\NL = \length\co\NL$. 
  We  make use of
  functions from the class $\tally\FAC^0 = \tally \cap \FAC^0$ which is
  contained in $\tally\NL$.

  Each language $L \subseteq
  \{ 0,1\}^{\ast}$ has an associated total \emph{characteristic function}
  $\chi_L \colon \{ 0,1\}^{\ast} \to \{0, 1\}$ defined  by $\chi_L (w) = 1$
  if and only if $w \in L$.  

  $\Parity \subseteq \{ 0,1\}^*$ is the set of binary strings that contain an odd number of 1s.

    \subsection{Reductions} %
    For concreteness, we explicitly define some standard types of reductions. 
    Let $A, B \subseteq \{ 0,1\}^{\ast}$.
  
    \begin{definition}[Many-one reducible]
      \label{def:many_One_reduction}
      Set $A$ is many-one reducible to set $B$, written $A \leq^{\C}_m B$,
      if there is a function $f$ that  is $\C$-computable with the
      property that for all $w$, $w \in A$, if and only if $f(w) \in B$.
    \end{definition}

    The following definition of truth table reductions comes
    from~\cite{BB1988,BHL1995}, for a more formal definition see~\cite{LLS1975}.
  
    \begin{definition}[Truth-table reduction]
      \label{def:truthtablereduction}
      Set $A$ is $\C$ truth-table reducible to $B$,
      written $A \leq^{\C}_{tt}B$,
      if there exists $\C$-computable functions $\tau$
        and $\sigma$ %
        such that for all $w \in \{0,1\}^*$,
        $\tau(w)$ is 
          a list of $\ell \in \mathbb{N}$ strings $ a_1, \ldots, a_{\ell}$,
        also $\sigma(w)$ is a truth table (Boolean function) with $\ell$ variables,
        and $w \in A$ if and only if 
        $\sigma( \chi_B(a_1),\ldots, \chi_B(a_{\ell}) )$ evaluates to true,
        where
        $\chi_{B}$ is the characteristic function of $B$.
    \end{definition}
    A \emph{disjunctive} truth table reduction (dtt) is one where 
      at least one string generated by $\tau(w)$ is in~$B$.
    Or equivalently, where $\sigma(w) = \bigvee_{1\leq i \leq \ell} \chi_B(a_i)$.
    A \emph{conjunctive} truth table reduction (ctt) is one where 
      all the strings generated by $\tau(w)$ are in $B$.
    Or equivalently, where $\sigma(w) = \bigwedge_{1\leq i \leq \ell} \chi_B(a_i)$.

    \begin{definition}[Turing reducible]
      \label{def:turing_reduction}
      Set $A$ is $\C$ Turing reducible to $B$, written $A \leq^{\C}_{T}B$, if
      there is a $\C$-computable oracle circuit (or Turing machine) $M$ such that   $w \in A$
      iff $M$ accepts $w$ with $B$ as its oracle. %
    \end{definition}

    The following implications follow directly from these definitions, 
      for more details see~\cite{LLS1975}.
      \begin{equation*}
        \begin{tikzpicture}[baseline]
          \matrix (m) [row sep=-1em, column sep=-0.8em]{
            & \node[rotate=15]{$\implies$}; & \node{$A \leq^{\C}_{dtt} B$}; & \node[rotate=-15]{$\implies$}; & \\ %
             \node{$A \leq^{\C}_{m} B$}; & & & & \node{$A \leq^{\C}_{tt} B \implies A \leq^{\C}_T B$ };\\ %
            & \node[rotate=-15]{$\implies$}; & \node{$A \leq^{\C}_{ctt} B$}; & \node[rotate=15]{$\implies$}; & \\
          };
        \end{tikzpicture}
      \end{equation*}

    Let $\Reducible{r}{\FAC^0}{\C}$ be the set of all languages that are
      $\FAC^0$ reducible to languages in $\C$ via some type of reduction 
      $\mathrm{r} \in \{ \text{m, dtt, ctt, tt, T} \}$. 

  \subsection{Some useful $\FAC^0$ functions}
    \label{sec:facfunc}
    \paragraph{Pairing function}
    We require a pairing function that is injective and extremely easy
    ($\FAC^0$) to compute.  We use the pairing function
    that interleaves the bits of two binary string arguments $a$ and $b$. For
    example, the binary strings $a=a_2 a_1 a_0$ and $b=b_2 b_1 b_0$ are paired
    as the interleaved string $\langle a,b \rangle =
    b_2 a_2 b_1 a_1 b_0 a_0$.
    The circuits for interleaving and de-interleaving have only  a single input gate layer and
    a single output gate layer (and so are 2-layer $\AC^0$ circuits). This circuit can be shown to be
    $\DLOGTIME$-uniform.

    \paragraph{Binary to Unary}
    There is a 
      constant depth
        circuit family where circuit $C_n$ 
        takes as input some word  $w\in \{0,1\}^n$ 
        and outputs $1^x$ where $x$ is the positive integer encoded in the 
        first $\lceil\log_2 n \rceil$ bits of $w$~\cite{CSV1984}. 
    It can be shown that this circuit family is  
      $\DLOGTIME$ uniform and so this conversion from short binary strings to unary is in $\FAC^0$. 

    \paragraph{Unary to Binary}
    There is a 
      constant depth
        circuit family where circuit $C_n$ 
        takes as input some word $w = 0^{n-x} 1^x$ where $0 \leq x \leq n $,
        and outputs the binary encoding of $x$~\cite{CSV1984}. 
    It can be shown that this circuit family is $\DLOGTIME$ uniform and so 
      unary to binary conversion is in $\FAC^0$.
 
    \subsection{Configuration graphs}
    \label{sec:config_graphs}
    \begin{definition}[Configuration Graph]
     Let $w \in \{0,1\}^*$ be the input to a halting Turing machine~$M$. 
      The \emph{configuration graph} $C_{M,w}$ is a directed acyclic graph 
        where each vertex encodes a configuration of $M$ 
        on inputs of length $|w|$.  
      The graph $C_{M,w}$ has a directed
          edge from a vertex~$c$ to a vertex $c'$ if the configuration encoded by $c'$ 
          can be reached from the configuration encoded by~$c$ in one
          step via $M$'s transition function.
    \end{definition}

    A configuration graph $C_{M,w}$ has the property that there is a directed
      path from the vertex $c_s$ representing the start configuration, to the accept vertex $c_a$ 
      if an only if $M$ accepts input~$w$. 
    Lemma~\ref{lem:config_in_AC0} follows from~\cite{Imm1987,Imm1999x}. 
    \begin{lemma}
      \label{lem:config_in_AC0}
      Given the binary encoding of 
        a Turing machine $M$, which has state set $Q$ and
        has an $\FAC^0$ computable space bound $s = \bigO{\log |w|}$,
        and given an input $w$, 
        the configuration graph $C_{M,w}$  
        is computable in $\uniform{\DLOGTIME}{\FAC^0}$
        and is  of size $\bigO{2^s |w| |Q|}$.
    \end{lemma}

  \section{Languages reducible to $\tally\NL$}
  \label{sec:tallynl}

  In this work we consider the class $\tally\NL$ as well as classes  $\AC^0$
  many-one, disjunctive truth-table, conjunctive truth-table, truth-table, and
  Turing reducible to $\tally\NL$. Their containment relationships are shown in
  Figure~\ref{fig:summaryofallresults}. We prove the following for
  completeness.

  \begin{figure}%
    \centering
    \includegraphics{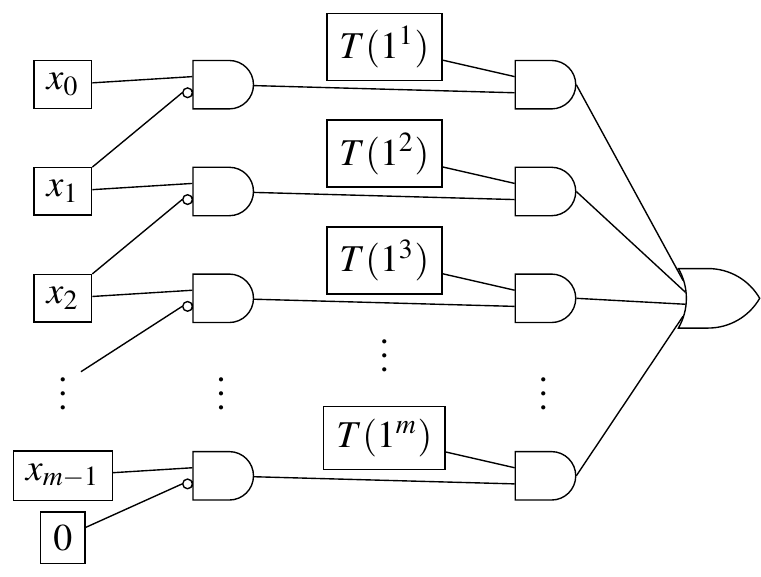}
    \caption{A gadget that simulates a single tally oracle gate. Gates of the
      form $T(1^i)$ are constant gates that simulate a Turing machine $T$:
      where $T(1^i)=1$ if the Turing machine $T$ accepts input $0^{m-i}1^i$,
      and $T(1^i)=0$ otherwise.}
    \label{fig:tally-gate-gadget}
  \end{figure}

  \begin{lemma}
    \label{lem:T_reduce_tallyNL_subneq_NL}
    $\Reducible{T}{\FAC^0}{\tally\NL} \subsetneq \NL$ 
  \end{lemma}
  \begin{proof} ($\subseteq$) 
    Let $L \in \Reducible{T}{\FAC^0}{\tally\NL}$. 
    Since the circuit and the oracles compute
    functions in $\NL$, there is a non-deterministic logspace Turing machine
    that computes the composition of these functions. 

    ($\neq$) $\Parity \in \NL$. We know that
      $\Parity \not\in \textrm{non-uniform-}\AC^0$~\cite{FSS1984p} and
      that $\tally \subseteq \textrm{non-uniform-}\AC^0$,
      hence it is sufficient to prove that 
      $\Reducible{T}{\FAC^0}{\tally\NL} \subseteq \textrm{non-uniform-}\AC^0$.

    Let $L \in \Reducible{T}{\FAC^0}{\tally\NL}$.
    Consider a family of circuits $\mathcal{C}_L$
    that recognizes $L$ and makes use of the Turing machine~$\tallyMachine$ as the tally %
    oracle. Let $w \in \{0,1\}^\ast$, and consider the circuit $C_{|w|} \in
    \mathcal{C}_L$ that decides whether or not $w \in L$. There is some number
    $k \in \mathbb{N}$ of oracle gates in $C_{|w|}$. The $i$th
    such oracle gate, $i \in \{1,2, \ldots, k \}$, takes one of $m+1$ inputs
    where $m$ is the number of wires into the gate (recall that inputs to the
    gate are of the form $0^{m-j}1^j$).  We (non-uniformly) replace oracle
    gate $i$ with the gadget shown in Figure~\ref{fig:tally-gate-gadget}. 
    This gadget encodes tally machine answers as constants.   
        The replacement can be done knowing $|w|$ (and not knowing $w$). We
    replace all $k$ tally oracle gates with this gadget to get a new circuit that is
      a constant factor (i.e.\ 5 times) deeper than $C_{|w|}$ and polynomially (in $|w|$)
    larger. Applying this transformation to the entire family $\mathcal{C}$
    results in a non-uniform $\AC^0$ circuit family that recognizes $L$.
  \end{proof}

    The same proof gives 
         $\Reducible{T}{\FAC^0}{\tally} 
        \subseteq
      \textrm{non-uniform-}\AC^0$ and hence $\Reducible{T}{\FAC^0}{\tally}
        \neq
      \NL$, which holds for
    $\tally$ as opposed to $\tally\NL$, and also for Turing reductions that are
    uniform-$\FAC^0$, or non-uniform-$\FAC^0$.

  \section{Uniform $\ORgate$ circuits}
  \label{sec:or}
  In this section we consider the relationship between uniform polynomial-size
  $\ORgate$ circuits and $\tally\NL$. We also consider the classes of
  languages reducible to these classes by suitably weak reductions.
  We begin with a $\length\NL$ lower bound on uniform polynomial size $\ORgate$ circuits.
  For this lower bound we consider $\length\NL$ rather than $\tally\NL$ because $\ORgate$ circuits act on binary strings and $\length\NL$ is a binary analogue of $\tally\NL$ (with almost the same proof we get an analogous $\tally\NL$ lower bound for $ \uniformOR{\FAC^0}$ if we  restrict to inputs from $\{ 1 \}^*$).

  \begin{theorem}
    $\length\NL \subsetneq \uniformOR{\FAC^0}$.
    \label{thm:tallyNL-in-uni-OR}
  \end{theorem}
  \begin{proof}
    Let $L \in \length\NL$.
    $L$ is accepted by a non-deterministic logspace Turing machine~$\tallyMachine$,
      for which one or more computation paths are accepting exactly
      for those words $w \in L \subseteq \{ 0, 1 \}^*$.
    The configuration graph $C_{\tallyMachine,w}$ for~$\tallyMachine$ on input
    $w\in\{0,1\}^*$ 
      is $\FAC^0$ computable from $\tallyMachine$ and $w$ (see Lemma~\ref{lem:config_in_AC0}).
    We construct the configuration graph assuming that its input $w$ is $1^{|w|}$ (recall
    that if $w \in L$ then all words in $\{0,1\}^{|w|}$ are in $L$).
    We modify the graph $C_{\tallyMachine,w}$ to create an $\ORgate$ circuit as follows.
    Each edge becomes a wire and each vertex becomes an $\ORgate$ gate, 
      except the start vertex (representing the initial configuration of
      $\tallyMachine$ on input $1^{|w|}$)
      which becomes a constant 1 gate. We add $|w|$ ``dummy'' input gates that are not wired to anything.  
    We add a new $\ORgate$ gate that is the circuit's output gate,
      and a constant 0 is wired into the every $\ORgate$  gate in the circuit.
    All accept-vertices (representing the accepting configurations) are wired into this output gate.
    If $w \in L$ the circuit accepts since there is a path from 1 to the output gate.
    If $w \not\in L$ the circuit rejects since there is no path from 1 to the output gate.

    If we apply this transformation to the set of all configurations graphs 
      for the fixed machine $\tallyMachine$ over all inputs $w \in \{ 1\}^*$,
      we get a circuit family $\mathcal{C}$.
    Members of such a  circuit family are computable by an $\FAC^0$ function
    $f_\tallyMachine: \{1\}^* \rightarrow \mathcal{C}$. %
    
    Consider the language $L = \{ w \mid w \textrm{ has at least one 1} \} $ which is easily seen to be in $\uniformOR{\FAC^0}$ but not in $\length\NL$, giving the required inequality for strict containment.
  \end{proof}

  Next we show that the languages accepted by uniform polynomial-size  $\ORgate$
  circuits are strictly contained in those disjunctive truth-table reducible
  to $\tally\NL$.

  \begin{theorem}
    \label{thm:uniformOR_subneq_dttTallyNL}
    $\uniformOR{\FAC^0} \subsetneq \Reducible{dtt}{\FAC^0}{\tally\NL}$
  \end{theorem}
  \begin{proof}
    It is trivially the case that 
      $\uniformOR{\FAC^0} \subseteq \Reducible{m}{\FAC^0}{\uniformOR{\FAC^0}}$.
    Then, by applying Theorem~\ref{lem:theORcluster_eq} (stated and proved below) we get 
      that $\uniformOR{\FAC^0} \subseteq \Reducible{dtt}{\FAC^0}{\tally\NL} %
        = \linebreak \Reducible{m}{\FAC^0}{\uniformOR{\FAC^0}}$.
    To show strict containment, observe that $\Reducible{dtt}{\FAC^0}{\tally\NL}$
      contains languages in $\AC^0  \cap \, \textrm{non-uniform-}\CAND$ that
      are not accepted by any $\ORgate$ circuit family. 
  \end{proof}
    
  Since the previously stated upper and lower bounds on $\uniformOR{\FAC^0}$
  are both strict, it is natural to ask how $\uniformOR{\FAC^0}$ relates to the
  most obvious class that lies between these bounds, namely
  $\Reducible{m}{\FAC^0}{\tally\NL}$.
  In fact, we get an inequality:
  $\uniformOR{\FAC^0} \neq \Reducible{m}{\FAC^0}{\tally\NL}$, as
  $\Reducible{m}{\FAC^0}{\tally\NL}$ contains languages in $\AC^0  \cap\,  \textrm{non-uniform-}\CAND$
  that are not accepted by any $\ORgate$ circuit family. 

  The remainder of this section is concerned with the proof of
  Theorem~\ref{lem:theORcluster_eq}, which was used in
  Theorem~\ref{thm:uniformOR_subneq_dttTallyNL} to give an upper bound on
  $\uniformOR{\FAC^0}$, and shows the equivalence of three complexity
  classes. %

  \begin{theorem}
    \label{lem:theORcluster_eq}
    The following classes are equal:
    \begin{compactitem}
      \item $\Reducible{m}{\FAC^0}{\uniformOR{\FAC^0}}$
      \item $\Reducible{dtt}{\FAC^0}{\uniformOR{\FAC^0}}$ 
      \item $\Reducible{dtt}{\FAC^0}{\tally\NL}$
    \end{compactitem}
  \end{theorem}
  This theorem is proven by the inclusion cycle %
  in Lemmas 
    \ref{lem:mReduceUniOR_subset_dttUniformOR}, 
    \ref{lem:dttReduceUniOR_subset_dttTallyNL}, and
    \ref{lem:dttReduceTallyNL_subset_mUniformOR} below.

  \begin{lemma}
    \label{lem:mReduceUniOR_subset_dttUniformOR}
    $\Reducible{m}{\FAC^0}{\uniformOR{\FAC^0}} \subseteq 
     \Reducible{dtt}{\FAC^0}{\uniformOR{\FAC^0}}$
  \end{lemma}
  \begin{proof}
    The latter class is a generalization of the former.
  \end{proof}

  \begin{lemma}
    \label{lem:dttReduceUniOR_subset_dttTallyNL}
    $\Reducible{dtt}{\FAC^0}{\uniformOR{\FAC^0}} \subseteq 
     \Reducible{dtt}{\FAC^0}{\tally\NL}$
  \end{lemma}
  \begin{proof}
    Let $L \in \Reducible{dtt}{\FAC^0}{\uniformOR{\FAC^0}}$ with oracle
    language $L' \in \uniformOR{\FAC^0}$.  That is, there exists a function
    $\tau \in \FAC^0$ mapping from $\{0,1\}^*$ to the set of tuples of binary
    words where \emph{at least one} word in the tuple $\tau(w) = (x_1,
    x_2,\ldots,x_m)$ is in $L' $ iff $w \in L$.
    
    To show that any of the binary words $\tau(w) = (x_1, x_2,\ldots,x_m)$ are
    in~$L'$ (i.e.\ are accepted by the $\ORgate$ circuit family)  it is
    sufficient to show that there is a single bit 1 in a word from $\tau(w)$ such
    that the bit's assigned input gate is on a path to the output gate in the
    appropriate $\ORgate$ circuit (or that there is a constant 1 gate in some
    circuit that is on a path to the output gate).

    With this in mind, we define the function  $\tau' \in \FAC^0$, from
    $\{0,1\}^*$ to the set of tuples of unary words.  $\tau'(w) = (u_1,\ldots,
    u_{q(|w|)})$, where $q(|w|)$ is polynomial in $|w|$, such that for each
    bit $i$ in each word $x_l$ in $\tau(w)$, there is a unary word $u_{l,i}$ in
    $\tau'(w)$ that encodes both $|x_l|$ (i.e.\ the length of $x_l$) and $i$, specifically: 
    \begin{equation}
      \label{eq:tau_or}
      u_{l,i} = 
      \begin{cases}
        1^{\langle |x_l|, |x_l| \rangle} &
          \text{ if } i = |x_l|,\\
        1^{\langle i, |x_l| \rangle} & 
          \text{ if } 0 \leq i \leq |x_l|-1 \text{ and  bit } i \text{ of } x_l \text{ is } 1,\\
        1                            & 
          \text{ if } 0 \leq i \leq |x_l|-1 \text{ and  bit } i \text{ of } x_l \text{ is }  0.\\
      \end{cases}
    \end{equation}
    Here $u_{l,i}$ is the $(l,i)$th word in $\tau'(w)$, $x_l$ is the $l$th word
    in $\tau(w)$ and $\langle \cdot , \cdot \rangle$ denotes the pairing
    function in Section~\ref{sec:facfunc}.
    (Note that $0$ bits are not uniquely encoded; our construction does not require it.) 

    Now we argue that $\tau' \in \FAC^0$.
    Each of the $q(|w|)$ unary words in $\tau'(w)$ are computed independently and
    in parallel. %
    The $(l,i)$th unary word is
    computed as follows:
    First compute $x_l\in \{0,1\}^*$, which is the $l$th word in
    $\tau(w)$.  If the $i$th bit of $x_l$ is 0 then output the unary word 1.
    Otherwise compute the pairing $k = \langle i, |x_l| \rangle$
    (Section~\ref{sec:facfunc}), convert the binary number $k$ to unary to give $1^k$ 
    which is then output in an encoded form as $0^{z-k}1^{k}$ where $1 \leq k < z$, 
    $z = 2^{2\lceil \log |w| + 1 \rceil} \in \bigO{|w|^2}$.  The $(l,i)$th
    sub-circuit of $\tau'$ is composed of a constant number of $\FAC^0$
    computable routines from Section~\ref{sec:facfunc} along with the computation of $\tau$ which is, by hypothesis,
    in $\FAC^0$. The polynomial number $q(|w|)$ of such constant depth computations are done in parallel, hence $\tau' \in \FAC^0$. 

    Let $f \in \FAC^0$, $f: \{ 1\}^* \rightarrow \mathcal{C}$, be the
    uniformity function of the $\ORgate$-circuit family that recognises $L'$.
    We next define a non-deterministic Turing machine $\tallyMachine_f$ that takes unary
    input, and makes use of $f$. 
    The machine $\tallyMachine_f$ is defined to 
      reject on input word 1
      and accept
      input $1^{k}$ if $k>1$ and if the un-pairing (see Section~\ref{sec:facfunc})
      of the binary encoding of $k$ gives two binary
    numbers $n$ and $i$, such that \emph{there is a path} from the $i$th input
    gate to the output gate of circuit $f(1^n)$.
    $\tallyMachine_f$ also accepts if $i = n$ and there is a path
      from some constant 1 gate to the output gate of circuit $f(1^n)$.
    $\tallyMachine_f$ works as follows. $\tallyMachine_f$ computes the unary to binary conversion
    and the un-pairing routine in logspace (see Section~\ref{sec:facfunc}).
    By hypothesis, the uniformity function~$f $ is in $\FAC^0$
    so, by using the standard re-computation trick~\cite{AB2009x,Pap1993x}
    for logspace Turing machines,
       $\tallyMachine_f$ both computes~$f$ 
    and tests reachability from input gate $i$ to the output gate of
    circuit $f(1^n)$ in non-deterministic logspace.
    Hence,
      if there is a path from input gate $i$ (or some constant 1 gate) to the output gate
        then $\tallyMachine_f$ accepts,
      otherwise if no path is found 
        then $\tallyMachine_f$ rejects.
    Moreover, since $\tallyMachine_f$ uses space $O(\log k)$, the language it  accepts is in $\tally\NL$.

    $\tallyMachine_f$ will be our $\tally\NL$ oracle machine. 
    We now prove that for any $w \in \{0,1\}^*$,
        at least one word 
          in the tuple $\tau'(w)$
            is accepted by at least one of the $\tallyMachine_f$ oracle machines iff $w \in L$.
    If $w \in L$ then 
      there exists a word $x$ in the tuple $\tau(w)$ 
        with at least one
        bit with value 1 that is assigned to an input gate  
        that is on a path to the output gate in $\ORgate$ circuit $f(1^{|x|})$. 
    This means that the tuple of words  $\tau'(w)$
        contains at least one unary word that encodes $|x|$ and $i$, where $i$ is the bit position assigned to 1. 
    By the construction in the previous paragraph, this word in $\tau'(w)$ is accepted by~$\tallyMachine_f$. 

    If $w \notin L$ then by hypothesis there are no words in $\tau(w)$ that 
        are accepted by the uniform $\ORgate$ circuit family.
        Any $0$'s in words from $\tau(w)$ become encoded as the input 1 to $\tallyMachine_f$, which is rejected by $\tallyMachine_f$ since $k=1$.   
        While $\tau(w)$ may contain words $x$ with bits set to 1 (or constant bits set to~1),
        these bits are  assigned to input (or constant) gates that do not have a path to the output gate 
        in the circuit $f(1^{|x|})$.
     Hence, none of these words in $\tau'(w)$ will be accepted by the oracle calls to $\tallyMachine_f$. 
 
     Therefore $\tau'$ is a disjunctive truth-table reduction from $L$ to a language in $\tally\NL$. 
  \end{proof}

  \begin{lemma}
    \label{lem:dttReduceTallyNL_subset_mUniformOR}
     $\Reducible{dtt}{\FAC^0}{\tally\NL} \subseteq 
       \Reducible{m}{\FAC^0}{\uniformOR{\FAC^0}}$
  \end{lemma}
  \begin{proof}
    Let $L \in \Reducible{dtt}{\FAC^0}{\tally\NL}$ with  $T \in \tally\NL$ as the oracle language. 
    That is, there exists a function $\tau \in \FAC^0$ that maps
     $\{ 0,1 \}^*$ to 
      the set of tuples of unary words, where
     at least one word in the tuple $\tau(w) = (x_1, x_2,\ldots,x_\ell)$ is in $T $ iff $w \in L$.

    Let $r:\{0,1\}^* \to \{0,1\}^*$.
    Let the notation $r(w)_k$ denote the $k$th bit of the word $r(w)$.
    The function~$r$ is defined in a bitwise fashion as follows:
    \begin{equation}
      r(w)_k = 
      \begin{cases}
         1 & \text{if } 1^k \text{ is in the tuple } \tau(w),\\
         0 & \text{otherwise.}
      \end{cases}
    \end{equation}
    We claim that $r$ is an $\FAC^0$ many-one reduction from $L$ to a language in $\uniformOR{\FAC^0}$.

    First we prove that $r \in \FAC^0$.
    The circuit that computes $r(w)$ first computes the tuple $\tau(w)$,
      which is possible since $\tau \in \FAC^0$. 
    Without loss of generality we say
      that %
        $\tau(w)$ is a tuple of $\ell \in \mathbb{N}$ unary words,
          each of length $\leq q  \in \mathbb{N}$,
            and each of which is padded up to length $q$ with~$0$'s
            (i.e.\ the unary word $1^k$ is padded to be $0^{q-k}1^{k}$;
            this technicality comes from the fact that the circuit has
            a fixed number $q$ of wires used encode a unary string which is dependent on the circuit input).
    Then, in constant depth, the circuit translates each string of the form
      $0^{q-k}1^k$ into a string of the form $0^{q-k}1 0^{k-1}$.
    All $\ell$ such words are then bitwise $\ORgate$ed to give a single binary  
      string of length $q$, that represents $r(w)$. 
    This is all easily achieved in $\FAC^0$.

    We now describe a uniform polynomial-size $\ORgate$ circuit family $\mathcal{C}$. 
    Let $f_\tallyMachine : \{1\}^* \rightarrow \mathcal{C}$ be
      the uniformity function of the circuit family $\mathcal{C}$. 
    On $1^m$, the function $f_\tallyMachine$ creates  $m$ configuration graphs:
    one configuration graph $C_{\tallyMachine,k}$ of machine $\tallyMachine$
    (that accepts $T$) on input $1^k$ for each $k \in \{1,\ldots, m\}$ (a
    generalization of the technique used in the proof of
    Theorem~\ref{thm:tallyNL-in-uni-OR}).  Then, each of the $m$ graphs are
    modified and connected together to create a single $\ORgate$ circuit as
    follows.  Each edge becomes a wire.  The vertex in $C_{\tallyMachine,k}$
    that represents the start configuration of  $\tallyMachine$ on input $1^k$
    becomes the $k$th input gate of the $\ORgate$ circuit. 
    All other vertices become an $\ORgate$ gate.  For each $k$, all accept
    vertices of the graph $C_{\tallyMachine,k}$ (representing the accepting
    configurations) are wired into a new $\ORgate$ gate $o_k$.  We add a single
    constant 0 gate which is wired into every $\ORgate$ gate in the circuit.
    Finally each of the $o_k$ gates, where $1 \leq k \leq m$, are wired into a
    single $\ORgate$ gate which is the output gate. $\mathcal{C}$ is of
    polynomial size (each circuit $f_\tallyMachine(1^m)$ is of size polynomial
    in $m$), and it is relatively straightforward to verify that $\mathcal{C}$
    is $\FAC^0$ uniform. 

    We need to argue that the circuit family $\mathcal{C}$ accepts $r(w)$ iff
    $w \in L$.  Suppose $w \in L$. This implies that the tuple $\tau(w)$
    contains at least one word $1^j$ in the tally set $T$. In turn, this
    implies that bit $j$ in $r(w)$ is 1 (formally, $r(w)_j = 1$).
    Let $|r(w)| = m$. 
    The fact that $\tallyMachine$ accepts $1^j$ implies that the circuit $c_m =
    f_\tallyMachine(1^m) \in \mathcal{C}$ is constructed in such a way that its
    $j$th input gate is on a path to its output gate. Input gate $j$ is set to
    1, therefore circuit $c_m$ accepts $r(w)$.

    Suppose $w \not\in L$. Hence, no word in the tuple $\tau(w)$ is in the
    tally set $T$. 
    Let  $1^j$  be any unary word in the tuple $\tau(w)$.
    In turn, this implies that bit $j$ in $r(w)$ is 1 (formally, $r(w)_j = 1$).
    Let $|r(w)| = m$. Consider the circuit $C_m = f_\tallyMachine(1^m) \in
    \mathcal{C}$.  Since the Turing machine $\tallyMachine$ does not accept
    $1^j$, this implies that there is no path from input gate $j$ in $C_m$ to
    the output gate of $C_m$. 
    Since  $C_m$ is an $\ORgate$ circuit with no paths from the input gates that are
    set to 1 to the output gate, and where there are no constant 1 gates,  it
    rejects $r(w)$.

    Therefore $r$ is a many-one reduction from $L$ to a language in $\uniformOR{\FAC^0}$. 
  \end{proof}
  
  Section~\ref{sec:and} contains our results on $\ANDgate$ circuits, analogous
  to those shown here for $\ORgate$ circuits.

  We omit the proofs of the
  following theorems, which can be obtained using the techniques in this
  section and those in Section~\ref{sec:and}.

  \begin{theorem}
    \label{lem:theTuring_cluster_eq}
    The following classes are equal:
    \begin{compactitem}
      \item $\Reducible{T}{\FAC^0}{\uniformOR{\FAC^0}}$
      \item $\Reducible{T}{\FAC^0}{\uniformAND{\FAC^0}}$
      \item $\Reducible{T}{\FAC^0}{\tally\NL}$
    \end{compactitem}
  \end{theorem}
  
  \begin{theorem}
  \label{lem:theTruthTable_cluster_eq}
    The following classes are equal:
    \begin{compactitem}
      \item $\Reducible{tt}{\FAC^0}{\uniformOR{\FAC^0}}$
      \item $\Reducible{tt}{\FAC^0}{\uniformAND{\FAC^0}}$
      \item $\Reducible{tt}{\FAC^0}{\tally\NL}$
    \end{compactitem}
  \end{theorem}

  \section{Uniform $\ANDgate$ circuits}
  \label{sec:and}
  Here we give upper bounds and lower bounds on the power of uniform $\ANDgate$
  circuits in terms of $\tally\NL$ and problems reducible to $\tally\NL$.
  The proofs have a similar flow to those for $\ORgate$ circuits in the 
  Section~\ref{sec:or}, although in a number of cases different tricks are used. 
  
  We begin with an upperbound and lowerbound on polynomial-size uniform
  $\ANDgate$ circuits: i.e.\ the class $\uniformAND{\FAC^0}$.

  \begin{theorem}
    $\length\NL \subsetneq \uniformAND{\FAC^0}$.
    \label{thm:tallyNL-in-uni-AND}
  \end{theorem}
  \begin{proof}
    Let $L \in \length\NL$.
    Since $\length\NL = \length\co\NL$,  
      this implies that
    $L$ is accepted by a co-non-deterministic logspace Turing machine $\tallyMachine$,
      for which all computation paths are accepting exactly
      for those words $w \in L$. %
    The configuration graph $C_{\tallyMachine,w}$ for~$\tallyMachine$ on input $w\in\{0,1\}^*$ 
      is $\FAC^0$ computable from $\tallyMachine$ and $w$ (see Lemma~\ref{lem:config_in_AC0}).
    We construct the configuration graph assuming that its input $w$ is $1^{|w|}$ (recall
    that if $w \in L$ then all words in $\{0,1\}^{|w|}$ are in $L$).
    We modify the graph $C_{\tallyMachine,w}$ to create an $\ANDgate$ circuit as follows.
    Each edge becomes a wire and each vertex becomes an $\ANDgate$ gate, 
      except the start vertex (representing the initial configuration of $\tallyMachine$ on input $w$)
      which becomes a constant 0 gate. We add $|w|$ ``dummy'' input gates that are not wired to anything. 
    We add a new $\ANDgate$ gate that is the circuit's output gate,
      and a constant 1 is wired into every $\ANDgate$ gate in the circuit.
    All reject vertices (representing the rejecting configurations) are wired into the output gate.
    If $w \in L$ the circuit accepts since there is no path from 0 to the output gate.
    If $w \not\in L$ the circuit rejects since there is a path from 0 to the output gate.

    If we apply this transformation to the set of all configurations graphs 
      for the fixed machine $\tallyMachine$ over all inputs $w \in \{ 1\}^*$,
      we get a circuit family $\mathcal{C}$.
    Members of such a  circuit family are computable by an $\FAC^0$ function
    $f_\tallyMachine: \{1\}^* \rightarrow \mathcal{C}$. %
    
            Consider the language $L = \{ 1^n \mid n \in \mathbb{N}  \} $ which is easily seen to be in $\uniformAND{\FAC^0}$ but not in $\length\NL$, giving the required inequality for strict containment.    
  \end{proof}

  Next we show that languages accepted by uniform polynomial-size  $\ANDgate$
  circuits are strictly contained in those conjunctive truth-table reducible
  to $\tally\NL$.

  \begin{theorem}
    \label{lem:uniformAND_subneq_cttTallyNL}
     $\uniformAND{\FAC^0} \subsetneq \Reducible{ctt}{\FAC^0}{\tally\NL}$
  \end{theorem}
  \begin{proof}
    It is trivially the case that 
      $\uniformAND{\FAC^0} \subseteq \Reducible{m}{\FAC^0}{\uniformAND{\FAC^0}}$.
    Then, by applying Theorem~\ref{lem:theANDcluster_eq} (stated and proved below) we get 
      that $\uniformAND{\FAC^0} \subseteq \Reducible{ctt}{\FAC^0}{\tally\NL} %
        =  \Reducible{m}{\FAC^0}{\uniformAND{\FAC^0}}$.
    To show strict containment, observe that $\Reducible{ctt}{\FAC^0}{\tally\NL}$
      contains languages in $\AC^0  \cap \, \textrm{non-uniform-}\COR$ that
      are not accepted by any $\ANDgate$ circuit family. 
  \end{proof}

  We also get the following inequality:
  $\uniformAND{\FAC^0} \neq \Reducible{m}{\FAC^0}{\tally\NL}$, as
  $\Reducible{m}{\FAC^0}{\tally\NL}$ contains languages in $\AC^0  \cap\,  \textrm{non-uniform-}\COR$
  that are not accepted by any $\ORgate$ circuit family. 

  The remainder of this section is concerned with the proof of
  Theorem~\ref{lem:theANDcluster_eq}, which was used in
  Theorem~\ref{lem:uniformAND_subneq_cttTallyNL} to give an upper bound on
  $\uniformAND{\FAC^0}$, and shows the equivalence of three complexity
  classes. %

  \begin{theorem}
    \label{lem:theANDcluster_eq}
    The following classes are equal:
    \begin{compactitem}
      \item $\Reducible{m}{\FAC^0}{\uniformAND{\FAC^0}}$
      \item $\Reducible{ctt}{\FAC^0}{\uniformAND{\FAC^0}}$ 
      \item $\Reducible{ctt}{\FAC^0}{\tally\NL}$
    \end{compactitem}
  \end{theorem}
  This theorem is proven by the cycle of inclusions in Lemmas 
    \ref{lem:mReduceUniAND_subset_cttUniformAND}, 
    \ref{lem:cttReduceUniAND_subset_cttTallyNL}, and
    \ref{lem:cttReduceTallyNL_subset_mUniformAND} below.

  \begin{lemma}
    \label{lem:mReduceUniAND_subset_cttUniformAND}
    $\Reducible{m}{\FAC^0}{\uniformAND{\FAC^0}} \subseteq 
     \Reducible{ctt}{\FAC^0}{\uniformAND{\FAC^0}}$
  \end{lemma}
  \begin{proof}
    The latter class is a generalization of the former.
  \end{proof}

  \begin{lemma}
    \label{lem:cttReduceUniAND_subset_cttTallyNL}
    $\Reducible{ctt}{\FAC^0}{\uniformAND{\FAC^0}} \subseteq 
     \Reducible{ctt}{\FAC^0}{\tally\NL}$
  \end{lemma}
  \begin{proof}
    Let $L \in \Reducible{ctt}{\FAC^0}{\uniformAND{\FAC^0}}$ with oracle
    language $L' \in \uniformAND{\FAC^0}$.  That is, there exists a function
    $\tau \in \FAC^0$ mapping from $\{0,1\}^*$ to the set of tuples of binary
    words where \emph{all} words in the tuple $\tau(w) = (x_1,
    x_2,\ldots,x_m)$ are in $L' $ iff $w \in L$.
    
    To show that any of the binary words $\tau(w) = (x_1, x_2,\ldots,x_m)$ are
    not in $L'$ (i.e.\ are rejected by the $\ANDgate$ circuit family)  it is
    sufficient to show that there is a single bit 0 in a word from $\tau(w)$ such
    that the bit's assigned input gate is on a path to the output gate in the
    appropriate $\ANDgate$ circuit (or that there is a constant 0 gate in some
    circuit that is on a path to the output gate).

    With this in mind, we define the function  $\tau' \in \FAC^0$, from
    $\{0,1\}^*$ to the set of tuples of unary words.  $\tau'(w) = (u_1,\ldots,
    u_{q(|w|)})$, where $q(|w|)$ is polynomial in $|w|$, such that for each
    bit $i$ in each word $x_l$ in $\tau(w)$, there is a unary word $u_{l,i}$ in
    $\tau'(w)$ that encodes both $|x_l|$ (i.e.\ the length of $x_l$) and $i$, specifically: 
    \begin{equation}
      \label{eq:tau_and}
      u_{l,i} = 
      \begin{cases}
        1^{\langle |x_l|, |x_l| \rangle} &
          \text{ if } i = |x_l|,\\
        1^{\langle i, |x_l| \rangle} & 
          \text{ if } 0 \leq i \leq |x_l|-1 \text{ and  bit } i \text{ of } x_l \text{ is } 0,\\
        1                            & 
          \text{ if } 0 \leq i \leq |x_l|-1 \text{ and  bit } i \text{ of } x_l \text{ is } 1.\\
      \end{cases}
    \end{equation}
    Here $u_{l,i}$ is the $(l,i)$th word in $\tau'(w)$, $x_l$ is the $l$th word
    in $\tau(w)$ and $\langle \cdot , \cdot \rangle$ denotes the pairing
    function in Section~\ref{sec:facfunc}.
    (Note that $1$ bits are not uniquely encoded; our construction does not require it.)

    Now we argue that $\tau' \in \FAC^0$.
    Each of the $q(|w|)$ unary words in $\tau'(w)$ are computed independently and
    in parallel. %
    The $(l,i)$th unary word is
    computed as follows:
    First compute $x_l\in \{0,1\}^*$, which is the $l$th word in
    $\tau(w)$.  If the $i$th bit of $x_l$ is 1 then output the unary word 1.
    Otherwise compute the pairing $k = \langle i, |x_l| \rangle$
    (Section~\ref{sec:facfunc}), convert the binary number $k$ to unary to give $1^k$ 
    which is then output in an encoded form as $0^{z-k}1^{k}$ where $1 \leq k < z$,
    $z = 2^{2\lceil \log |w| + 1 \rceil} \in \bigO{|w|^2}$.  The $(l,i)$th
    sub-circuit of $\tau'$ is composed of a constant number of $\FAC^0$
    computable routines from Section~\ref{sec:facfunc} along with the computation of $\tau$ which is, by hypothesis,
    in $\FAC^0$. The polynomial number $q(|w|)$ of such constant depth computations are done in parallel, hence $\tau' \in \FAC^0$. 

    Let $f \in \FAC^0$, $f: \{ 1\}^* \rightarrow \mathcal{C}$, be the
    uniformity function of the $\ANDgate$-circuit family that recognises $L'$.
    We next define a non-deterministic Turing machine $\tallyMachine_f$ that takes unary
    input, and makes use of $f$. 
    The machine $\tallyMachine_f$ is defined to 
      accept on input word 1
      and reject
      input $1^{k}$ if $k>1$ and if the un-pairing (see Section~\ref{sec:facfunc})
      of the binary encoding of $k$ gives two binary
    numbers $n$ and $i$, such that \emph{there is a path} from the $i$th input
    gate to the output gate of circuit $f(1^n)$.
    $\tallyMachine_f$ also accepts if $i = n$ and there is a path
    from some constant 0 gate to the output gate of circuit $f(1^n)$.
    $\tallyMachine_f$ works as follows. $\tallyMachine_f$ computes the unary to binary conversion
    and the un-pairing routine in logspace (see Section~\ref{sec:facfunc}).
    By hypothesis, the uniformity function~$f $ is in $\FAC^0$
    so, by using the standard re-computation trick~\cite{AB2009x,Pap1993x}
    for logspace Turing machines
      and the un-reachability algorithm~\cite{Imm1988p,Sze1988p}
       $\tallyMachine_f$ both computes~$f$ 
    and tests non-reachability from input gate $i$ to the output gate of
    circuit $f(1^n)$ in non-deterministic logspace.
    Hence,
      if there is a path from input gate $i$ (or some constant 0 gate) to the output gate
        then $\tallyMachine_f$ rejects,
      otherwise if no path is found 
        then $\tallyMachine_f$ accepts.
    Moreover, since $\tallyMachine_f$ uses space $O(\log k)$, the language it  accepts is in $\tally\NL = \tally\co\NL$.

    $\tallyMachine_f$ will be our $\tally\NL$ oracle machine. 
    We now prove that for any $w \in \{0,1\}^*$,
        all words
          in the tuple $\tau'(w)$
            are accepted by the $\tallyMachine_f$ oracle machines iff $w \in L$.
    If $w \notin L$ then 
      there exists a word $x$ in the tuple $\tau(w)$ 
        with at least one
        bit with value 0 that is assigned to an input gate  
        that is on a path to the output gate in $\ANDgate$ circuit $f(1^{|x|})$. 
    This means that the tuple of words  $\tau'(w)$
        contains at least one unary word that encodes $|x|$ and $i$, where $i$ is the bit position assigned to 0.
    By the construction in the previous paragraph, this word in $\tau'(w)$ is rejected by $\tallyMachine_f$. 

    If $w \in L$ then by hypothesis there are no words in $\tau(w)$ that 
        are rejected by the uniform $\ANDgate$ circuit family.
        Any $1$'s in words from $\tau(w)$ become encoded as the input 1 to $\tallyMachine_f$, which is accepted by $\tallyMachine_f$ since $k=1$.   
        While $\tau(w)$ may contain words $x$ with bits set to 0 (or constant bits set to~0),
        these bits are not assigned to input (or constant) gates that have a path to the output gate 
        in the circuit $f(1^{|x|})$.
     Hence, none of the words in $\tau'(w)$ will be rejected by the oracle calls to $\tallyMachine_f$.  
 
     Therefore $\tau'$ is a conjunctive truth-table reduction from $L$ to a language in $\tally\NL$. 
  \end{proof}

  \begin{lemma}
    \label{lem:cttReduceTallyNL_subset_mUniformAND}
     $\Reducible{ctt}{\FAC^0}{\tally\NL} \subseteq 
       \Reducible{m}{\FAC^0}{\uniformAND{\FAC^0}}$
  \end{lemma}
  \begin{proof}
    Let $L \in \Reducible{ctt}{\FAC^0}{\tally\NL}$ with  $T \in \tally\NL$ as the oracle language.
    That is, there exists a function $\tau \in \FAC^0$ that maps
     $\{ 0,1 \}^*$ to 
      the set of tuples of unary words, where
      all words in the tuple $\tau(w) = (x_1, x_2,\ldots,x_\ell)$ are in $T  $ iff $w \in L$.

    Let $r:\{0,1\}^* \to \{0,1\}^*$.
    Let the notation $r(w)_k$ denote the $k$th bit of the word $r(w)$.
    The function~$r$ is defined in a bitwise fashion as follows:
    \begin{equation}
      r(w)_k = 
      \begin{cases}
         0 & \text{if } 1^k \text{ is in the tuple } \tau(w),\\
         1 & \text{otherwise.}
      \end{cases}
    \end{equation}
    We claim that $r$ is an $\FAC^0$ many-one reduction from $L$ to a language in $\uniformAND{\FAC^0}$.

    First we prove that $r \in \FAC^0$.
    The circuit that computes $r(w)$ first computes the tuple $\tau(w)$,
      which is possible since $\tau \in \FAC^0$. 
    Without loss of generality we say
      that %
        $\tau(w)$ is a tuple of $\ell \in \mathbb{N}$ unary words,
          each of length $\leq q  \in \mathbb{N}$,
            and each of which is padded up to length $q$ with~$0$'s
            (i.e. the unary word $1^k$ is padded to be $0^{q-k}1^{k}$;
            this technicality comes from the fact that the circuit has
            a fixed number $q$ of wires used encode a unary string which is dependent on the circuit input).
    Then, in constant depth, the circuit translates each string of the form
      $0^{q-k}1^k$ into a string of the form $1^{q-k}0^11^{k-1}$.
    All $\ell$ such words are then bitwise $\ANDgate$ed to give a single binary  
      string of length $q$, that represents $r(w)$. 
    This is all easily achieved in $\FAC^0$.

    We now describe a uniform polynomial-size $\ANDgate$ circuit family $\mathcal{C}$. 
    Let $f_\tallyMachine : \{1\}^* \rightarrow \mathcal{C}$ be
      the uniformity function of the circuit family $\mathcal{C}$. 
    On $1^m$, the function $f_\tallyMachine$ creates  $m$ configuration graphs:
    one configuration graph $C_{\tallyMachine,k}$ of machine $\tallyMachine$
    (that accepts $T$) on input $1^k$ for each $k \in \{1,\ldots, m\}$ (a
    generalization of the technique used in the proof of
    Theorem~\ref{thm:tallyNL-in-uni-AND}).  Then, each of the $m$ graphs are
    modified and connected together to create a single $\ANDgate$ circuit as
    follows.  Each edge becomes a wire.  The vertex in $C_{\tallyMachine,k}$
    that represents the start configuration of  $\tallyMachine$ on input $1^k$
    becomes the $k$th input gate of the $\ANDgate$ circuit. 
    All other vertices become an $\ANDgate$ gate.  For each $k$, all reject
    vertices of the graph $C_{\tallyMachine,k}$ (representing the rejecting 
    configurations) are wired into a new $\ANDgate$ gate $o_k$.  We add a single
    constant 1 gate which is wired into every $\ANDgate$ gate in the circuit.
    Finally each of the~$o_k$ gates, where $1 \leq k \leq m$, are wired into a
    single $\ANDgate$ gate which is the output gate. $\mathcal{C}$ is of
    polynomial size (each circuit $f_\tallyMachine(1^m)$ is of size polynomial
    in $m$), and it is relatively straightforward to verify that $\mathcal{C}$
    is $\FAC^0$ uniform. 

    We need to argue that the circuit family $\mathcal{C}$ accepts $r(w)$ iff
    $w \in L$.  Suppose $w \notin L$. This implies that the tuple $\tau(w)$
    contains at least one word $1^j$ not in the tally set $T$. In turn, this
    implies that bit $j$ in $r(w)$ is 0 (formally, $r(w)_j = 0$).
    Let $|r(w)| = m$. 
    The fact that $\tallyMachine$ rejects $1^j$ implies that the circuit $c_m =
    f_\tallyMachine(1^m) \in \mathcal{C}$ is constructed in such a way that its
    $j$th input gate is on a path to its output gate. Input gate $j$ is set to
    0, therefore circuit $c_m$ rejects $r(w)$.

    Suppose $w \in L$. Hence, all words in the tuple $\tau(w)$ are in the
    tally set $T$. 
    Let  $1^j$  be any unary word in the tuple $\tau(w)$.
    In turn, this implies that bit $j$ in $r(w)$ is 0 (formally, $r(w)_j = 0$).
    Let $|r(w)| = m$. Consider the circuit $C_m = f_\tallyMachine(1^m) \in
    \mathcal{C}$.  Since the Turing machine $\tallyMachine$ does not reject
    $1^j$, this implies that there is no path from input gate $j$ in $C_m$ to
    the output gate of $C_m$. 
    Since  $C_m$ is an $\ANDgate$ circuit with no paths from the input gates that are
    set to 0 to the output gate, and where there are no constant 0 gates,  it
    accepts $r(w)$.

    Therefore $r$ is a many-one reduction from $L$ to a language in $\uniformAND{\FAC^0}$. 
  \end{proof}

  \section{Semi-uniform circuit families}
  \label{sec:uni-semi}
    We introduce a definition of semi-uniform families of Boolean circuits.
    This definition is inspired by the concept in membrane
    systems~\cite{PRRW2009x}.  Polynomial-size semi-uniform $\ORgate$ circuits,
    and $\ANDgate$ circuits, are shown to characterize~$\NL$.

    \begin{definition}[Semi-uniform circuit family]
      \label{def:semiuniformCircuitFam}
      A \emph{semi-uniform circuit family}~$\mathcal{C}$ is a set of
      Boolean circuits, each with a single output gate and no input gates, such
      that there is a function~$\semiencoder : \{ 0,1\}^\ast \rightarrow
      \mathcal{C}$ (computable within some resource bound)
      where~$\semiencoder(x) = C_x$.  We say that a semi-uniform circuit family~$\mathcal{C}$ 
      decides a language~$X$ if for each~$x$,
      the circuit~$\semiencoder(x) = C_x \in \mathcal{C}$ evaluates to~$1$ if~$x \in X$ and~$0$
      if~$x \notin X$.
    \end{definition}

    Here, $h$ is called the semi-uniformity function of $\mathcal{C}$. The
    intuition behind the definition is that the semi-uniformity function has
    access to the entire input word, whereas more standard uniformity functions
    access only the input word length (in unary).

    \begin{definition}[$\FOSOR$]
      Let $\FOSOR$ be the set of decision problems over a binary
      alphabet that are solved by $\FAC^0$ semi-uniform families of $\ORgate$
      circuits. 
    \end{definition}
    $\FOSAND$ is defined analogously using $\ANDgate$ circuits.
    Finally, the class\\ $\semiunif{\FAC^0}{\CANDCOR}$ is defined analogously
    using circuits that have both $\ANDgate$ and $\ORgate$ gates. 
    The proof of the following lemma is straightforward. 
    \begin{lemma}
      \label{thm:semiuniformCircuitsP}
      $\semiunif{\FAC^0}{\CANDCOR}=\P$
    \end{lemma}
    \begin{proof} 
      Any problem in $\P$ has a circuit family
        $\mathcal{C}$ with circuits using $\ANDgate$, $\ORgate$, and $\NOTgate$ gates 
          that is uniform by some function $f \in \FAC^0$, $f : \{ 1\}^* \rightarrow \mathcal{C}$.
      There is a semi-uniformity function
        $f' : \{ 0,1\}^* \rightarrow \mathcal{C}'$ 
          for a semi-uniform circuit family $\mathcal{C}'$,
            that simulates $f$  in the following way:
      For all $x \in \{0,1\}^*$, $f'(x)$ produces a circuit
        without input gates and where the string $x$ and its bitwise complement are available as constants,
        and the circuit carries out a dual-rail logic simulation~\cite{Gol1977p,GHR1995x} of the circuit $f(|x|)$.  
    \end{proof}

    \begin{lemma}
      \label{lem:semi-mono-or-NL}
      $\FOSOR = \NL$.
    \end{lemma}
    \begin{proof} 
      ($\NL \subseteq \FOSOR$)
      Let $L \in \NL$.
      $L$ is accepted by a non-deterministic logspace Turing machine $M$,
        i.e.\ one or more computation paths are accepting exactly for 
        those words $w \in L \subseteq \{ 0,1 \}^*$.
      Consider the configuration graph $C_{M,w}$ for~$M$ on input $w\in\{0,1\}^*$,
        which is $\FAC^0$ computable from $M$ and $w$ (see Section~\ref{sec:config_graphs}).
      We modify the graph $C_{M,w}$ to create an $\ORgate$ circuit as follows.
      Each edge becomes a wire and each vertex becomes an $\ORgate$ gate, 
        except the start vertex (which represents the initial configuration of $M$ on $w$)
        which becomes a constant~$1$ gate.
      All {\em accepting vertices} (representing  accepting configurations) are also wired to this output gate.
      We add a single constant 0 gate which is wired into every $\ORgate$ gate in the circuit.
      If $w \in L$ the circuit accepts since there is a path from 1 to the output gate.
      If $w \not\in L$ the circuit rejects since there is no path from 1 to the output gate and a 0 feeds into that gate. 
      These simple modifications can be made in $\FAC^0$.
      
      Fixing the machine $M$, and then considering this transformation on the set of all 
      configurations graphs, one for each input $w \in \{0, 1\}^*$,
        we get a semi-uniform circuit family~$\mathcal{C}$.
      Members of such a semi-uniform circuit family are computable by an $\FAC^0$ function
      $f_M: \{0,1\}^* \rightarrow \mathcal{C}$. %

      ($\FOSOR \subseteq \NL$) 
      Let $\mathcal{C}$ be a semi-uniform $\ORgate$ circuit family that
      recognizes $L \in \FOSOR$, we claim that there is a non-deterministic
      logspace Turing machine $M$ that recognizes~$L$. Let $h: \{ 0,1 \}^\ast
      \rightarrow \mathcal{C}$ be the semi-uniformity function of
      $\mathcal{C}$. On input $x \in \{ 0,1 \}^\ast$, $M$ computes $h(x)$ and
      performs a simple reachability on the resulting $\ORgate$ circuit in the
      following way: $M$ guesses a gate, if that gate is a constant 1-gate $M$ 
      then guesses a path from that gate, if the
      path ends at the output gate $M$ accepts. 
    \end{proof}

    \begin{lemma}
      \label{lem:semi-mono-and-NL}
      $\FOSAND = \NL$.
    \end{lemma}
    \begin{proof} 
      ($\NL \subseteq \FOSAND$)
      Let $L \in \tally\NL$.
      Since $\tally\NL = \tally\co\NL$ (Lemma~\ref{lem:tallynl_eq_tallyconl}),
        this implies that
          $L$ is accepted by a co-non-deterministic logspace Turing machine $M$,
          for which all computation paths accept exactly for those words $w \in L \subseteq \{ 0,1 \}^*$.
      Consider the configuration graph $C_{M,w}$ for~$M$ on input $w\in\{ 0, 1\}^*$,
        which is $\FAC^0$ computable from $M$ and $w$ (see Section~\ref{sec:config_graphs}).
      We modify the graph $C_{M ,w}$ to create an $\ANDgate$ circuit as follows.
      Each edge becomes a wire
        and each vertex becomes an $\ANDgate$ gate, 
        except the start vertex (which represents the initial configuration on $M$ on $w$) 
          which becomes a constant 0 gate. 
      We add a new $\ANDgate$ gate that is the circuit's output gate.
      All reject vertex (representing the reject configurations) are wired into this output gate.
      We add a single constant 1 gate which is wired into every $\ANDgate$ gate in the circuit.
      These modifications can be made in $\FAC^0$.
      If $w \in L$ the circuit accepts since there is no path from 0 to the output gate.
      If $w \not\in L$ the circuit rejects since there is a path from 0 to the output gate. 

      Fixing the machine $M$, and then considering this transformation on the set of all 
      configurations graphs, one for each input $w \in \{0, 1\}^*$,
        we get a semi-uniform circuit family~$\mathcal{C}$.
      Members of such a semi-uniform circuit family are computable by an $\FAC^0$ function
      $f_M: \{0,1\}^* \rightarrow \mathcal{C}$. %

      ($\FOSAND \subseteq \NL$) 
      Let $\mathcal{C}$ be a semi-uniform $\ANDgate$ circuit family that
      recognizes $L \in \FOSAND$. We claim that there is a co-nondeterministic
      logspace Turing machine $M$ that recognizes~$L$ and thus $ L \in  \NL$. Let $h: \{ 0,1 \}^\ast
      \rightarrow \mathcal{C}$ be the semi-uniformity function of
      $\mathcal{C}$. On input $x \in \{ 0,1 \}^\ast$, $M$ computes $h(x)$ and
      performs a simple reachability on the resulting $\ANDgate$ circuit in the
      following way. Starting at the output gate, $M$ guesses a path along the
      reverse direction of the edges (wires) until the path terminates. If the
      path terminates at a constant 1 gate $M$ accepts, otherwise $M$ rejects 
      (in the latter case the path terminates at a 0 gate, as by definition there are no AND gates with 
      in-degree 0 in the circuit). $M$ accepts $x$ if and only if all of its 
      computations accept, which is equivalent to saying that each path from
      an in-degree 0 gate to the circuit's output gate begins at a constant 1 gate,
      and so the circuit accepts. 
    \end{proof}

    We have the following separation between uniform polynomial-size  and semi-uniform $\ORgate$  circuits. The result also holds for $\ANDgate$ circuits.
    \begin{theorem}
      \label{cor:uniNotEqSemi}
    \hspace{1ex}  %
    \begin{compactitem} 
      \item   $\FOOR \subsetneq \FOSOR$  
     \item $\FOAND \subsetneq \FOSAND$
        \end{compactitem}
    \end{theorem}
\begin{proof} Follows from Theorem~\ref{lem:T_reduce_tallyNL_subneq_NL} and the containments in Figure~\ref{fig:summaryofallresults}.
\end{proof}

  \begin{comment}
    To prove this we first observe that $\FOOR \subseteq \FOSOR$, 
      as the former is a restriction of the latter.
    Also, Lemma~\ref{lem:semi-mono-or-NL} states that $\FOSOR = \NL$,
      hence we know that $\Parity \in \FOSOR$.
    However, the following lemma shows that $\Parity \notin \FOOR$.
         
    \begin{lemma}
      $\Parity \notin \FOOR$ 
     \label{lem:parity_not_in_ormax}
    \end{lemma}
    \begin{proof}
      For any uniform $\ORgate$ circuit family there exists a non-uniform
      $\ORgate$ circuit family of size and depth 1 that computes the same
      function, and since $\Parity \notin \textrm{non-uniform-}\AC^0$ this
      implies that $\Parity \notin \FOOR$.
    \end{proof}

    Analogously for $\ANDgate$ circuits:
    \begin{lemma}
      $\Parity \notin \FOAND$ 
     \label{lem:parity_not_in_and}
    \end{lemma}
    We know that $\FOAND \subseteq \FOSAND$
      (the former is a restriction of the latter),
        $\FOSAND = \NL$ (Lemma~\ref{lem:semi-mono-and-NL}), 
      and $\Parity \in \NL$, so we get:
    \begin{corollary}
      \label{cor:uniNotEqSemANDi}
      $\FOAND \subsetneq \FOSAND$
    \end{corollary}
  \end{comment}

  \section*{Acknowledgements}
    Many thanks to Eric Allender for valuable comments and discussion on
    uniform Boolean circuits and complexity classes within $\P$. We also thank
    Antonio E.\ Porreca,  David Doty, Jack Lutz and Dirk Walther for interesting discussions. 
    
  \appropriatebib
  \bibliography{diffuni}

\end{document}